\documentclass[a4paper,twocolumn,superscriptaddress,11pt,accepted=2019-05-07]{quantumarticle}
\pdfoutput=1
\usepackage[utf8]{inputenc}
\usepackage[english]{babel}
\usepackage[T1]{fontenc}
\usepackage{amsmath}
\usepackage{hyperref}

\usepackage{tikz}

\usepackage[numbers,sort&compress]{natbib}

\usepackage{amsfonts, amsmath, amssymb, graphicx, hyperref, color}

\usepackage[english]{babel}
\usepackage{amsfonts}				
\usepackage{amsmath}
\usepackage{bbm}
\usepackage{bm}
\usepackage{mathtools}
\usepackage{xcolor}
\usepackage{graphicx}
\usepackage{cleveref}
\usepackage{amsthm}
\usepackage{lipsum}

\newcommand{\ket} [1] {\vert #1 \rangle}
\newcommand{\bra} [1] {\langle #1 \vert}

\newcommand{\proj}[1]{\ket{#1}\bra{#1}}

\def\be{\begin{equation}}
\def\en#1{\label{#1}\end{equation}}

\def\bar#1{\overline #1}

\newtheorem{thm}{Theorem}

\newtheorem{lem}[thm]{Lemma}

\newtheorem{coro}[thm]{Corollary}

\begin{document}

\title{Simulating boson sampling in lossy architectures}
\date{\today}

\author{Ra\'{u}l Garc\'{i}a-Patr\'{o}n}
\affiliation{Centre for Quantum Information and Communication, Ecole Polytechnique de Bruxelles, CP 165, Universit\'{e} Libre de Bruxelles, 1050 Brussels, Belgium}
\email{raulgarciapatron@gmail.com}
\orcid{0000-0003-1760-433X}
\author{Jelmer J. Renema}
\affiliation{Clarendon Laboratory, Department of Physics, University of Oxford, Oxford OX1 3PU, United Kingdom}
\affiliation{University of Twente, PO Box 217, 7500 AE Enschede, The Netherlands}
\author{Valery Shchesnovich}
\affiliation{Centro de Ci\^{e}ncias Naturais e Humanas, Universidade Federal do ABC, Santo Andr\'{e}, SP, 09210-170 Brazil.}

\maketitle

\begin{abstract}
Photon losses are among the strongest imperfections affecting multi-photon interference. Despite their importance, little is known about their effect on boson sampling experiments. In this work we show that using classical computers, one can efficiently simulate multi-photon interference in all architectures that suffer from an exponential decay of the transmission with the depth of the circuit, such as integrated photonic circuits or optical fibers. We prove that either the depth of the circuit is large enough that it can be simulated by thermal noise with an algorithm running in polynomial time, or it is shallow enough that a tensor network simulation runs in quasi-polynomial time. This result suggests that in order to implement a quantum advantage experiment with single-photons and linear optics new experimental platforms may be needed. 
\end{abstract}

\section{Introduction}

\begin{figure}[t!]
\centering
  \includegraphics[width=.9\linewidth]{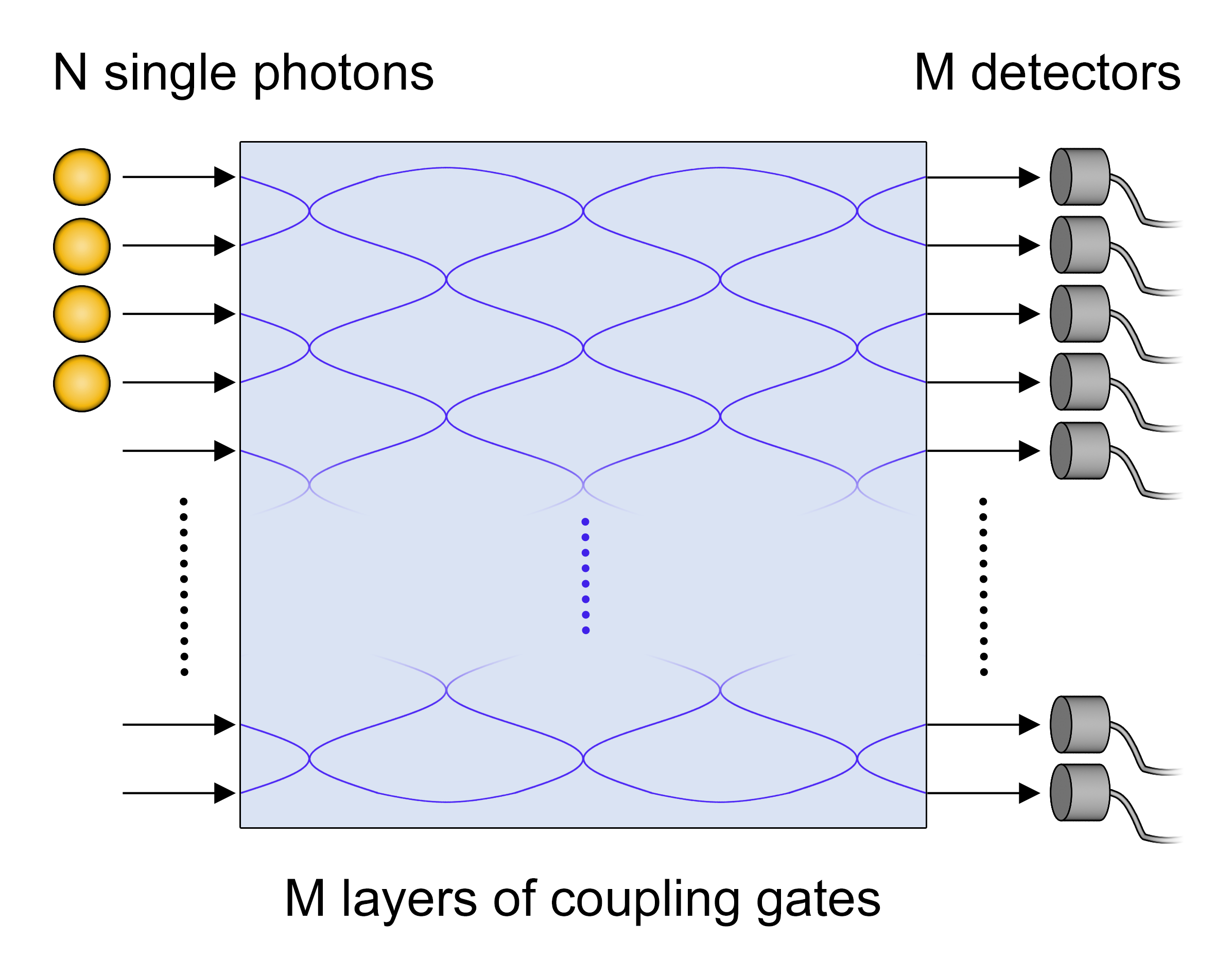}
\caption{In a boson sampling device $N$ single photons are sent over an $M$ mode linear optical circuit composed of $M$ layers of two-mode coupling gates and detected at the output with photon counting detectors. The circuit being selected randomly from the Haar measure,
we can place the $N$ photons over the first $N$ modes without lost of generality.}
\label{fig:figure1}
\end{figure} 

In 2003 Knill, Laflamme and Milburn showed that single-photon sources and linear optics are sufficient to achieve universal quantum computation \cite{Knill2001}.
A single-photon and linear optics version of measurement based quantum computation has also been thoroughly studied \cite{Kok2007,Rudol2016}.
In both proposals, a key component to reach universality is the capability for a measurement outcome to change the gates implemented later in time,
i.e., using active feed-forward, a challenging experimental requirement  \cite{Prev2007}.
In 2010, Aaronson and Arkhipov demonstrated that removing the feedforward condition provides a framework, called \textit{boson sampling}, that does not seem to be
sufficient for universal quantum computation but is hard to simulate on a classical computer \cite{Aaronson2013}.
The key idea behind their proof is the connection between the output statistics of non-interacting indistinguishable photons
and the permanent \cite{Caia53}, a quantity know to be \#P-hard to compute \cite{Vali79,Aaro11}.
As shown in Figure \ref{fig:figure1}, a boson sampling device implements an interference of $N$  single photons over a randomly-selected  $M$-mode linear optical interferometer and measures each output mode with a photon counting detector.
An $M$-mode linear optical interferometer can be build out of two-mode couplers (beamsplitters) acting on neighboring modes and single-mode phases gates. In order to generate an arbitrary linear optical circuit, a depth of $M$ is needed \cite{Reck94,Clements2016}.

Since 2013, a variety of experimental quantum optics groups have implemented proof-of-principle implementations based on different architectures, such as reconfigurable integrated photonic circuits \cite{BSintegrated,Broome2013,Tilmann2013}, fiber-loops \cite{BSfiberloop}, 3D waveguides \cite{BS3Dwaveguide} or multimode fibers \cite{Defi16}, that have the potential of being scalable and therefore are candidates for a quantum advantage demonstration.
The motivation to further simplify the experimental scheme led to the proposal of scattershot boson sampling, a sampling problem as hard to simulate as the initial boson sampling proposal. Scattershot boson sampling solves the problem of obtaining $N$ single photons from state of the art probabilistic single-photon sources by using $M$ heralded two-mode squeezed vacuum state sources,
one per input mode of the boson sampling circuit \cite{Lund14}.

The lack of fault-tolerant error correction in quantum advantage architectures, such as in boson sampling experiments, implies that increasing the size and depth of the circuit would ultimately lead to a system that is equivalent to sampling random noise. Therefore, the existence of an opportunity window where noise has not yet destroyed
the quantum advantage but a classical algorithm, such as \cite{Nevi17,Cliff17}, can no longer simulate the system is fundamental for a conclusive quantum advantage experiment.
It is therefore of paramount importance to have a good understanding of when noise makes a quantum advantage architecture classically simulatable.
Rahimi-Keshari et al. provided in \cite{Kesh16} a first rigorous bound, which required both losses and dark counts of the detectors.
Unfortunately, this bound is independent of the size of the system and can not provide an answer in terms of losses independently of additional noise (dark-counts of  detectors).

Together with the indistinguishably of photons, for which an algorithm to simulate partially distinguishable photons was demonstrated recently in \cite{Renema2016}, losses are the most damaging imperfections challenging boson sampling.  Despite its importance, little is know about the effect of losses. Firstly, it was proven in \cite{Aaronson2016}  
that boson sampling with losses remains hard in the regime of a constant number of photons lost, a rather limiting assumption. Secondly, modeling a lossy circuit as a larger lossless one with additional environmental modes, it is rather straightforward to see that the algorithm of Clifford and Clifford \cite{Cliff17} can be extended to lossy architectures with only a constant overhead. This has two implications: (i) if an ideal boson sampling circuit with $N$ input photons can be simulated (with an exponential-time algorithm), the same can be achieved for arbitrary losses; (ii) lossy multi-photon interference becomes classically simulatable when the number of photons left scales as $O(\log N)$.
In section \ref{sec:mainresult} of this work we will considerably improve over 
this trivial corollary of \cite{Cliff17} by proving that if the numbers of photons that reach the detectors are less than $O(\sqrt{N})$, a boson sampling experiments becomes classically efficiently simulatable.

Most experimental boson sampling architectures, and all for which interference has been shown for more than 2 photons, are based on a planar geometry
of depth proportional to the number of input systems, where the loss per coupler in the circuit is constant, leading to a law of exponential decay of the
transmission with the depth of the circuit. In section \ref{sec:mainresult} of this work we show that for those platforms, and platforms which have similar
exponential decay, boson sampling experiments can be efficiently simulated classically. Therefore, we believe that for single-photons and linear optics to remain competitive
in the race for a quantum advantage demonstration new ideas are needed.
More precisely, we show that for those platforms
either the depth of the circuit ($D$) is large enough ($D\geq \text{O}(\log M)$)
that it can be simulated by thermal noise with an algorithm running in polynomial time,
or the depth of the circuit is short enough ($D\leq \text{O}(\log M)$) that a tensor network simulation, similar in spirit to \cite{Temme2012},
runs in quasi-polynomial time.

Not all optical architectures suffer from an exponential decay of the transmission, for example  free-space optics has a quadratic decay of transmission.
In section \ref{sec:extra} we extend the validity of the thermal noise simulation to this family of architectures, to scattershot boson sampling and boson sampling architectures
where the photon-counting detectors are replaced by Gaussian measurements \cite{Lund17,Chak17}.

In sections \ref{sec:thermalstates}, \ref{sec:tensornetwork} 
of the manuscript we provide the detailed proofs stated in section \ref{sec:mainresult} and
in section \ref{sec:nonuniformlosses} we generalize the result to non-uniform losses. 
Finally, we conclude in section \ref{sec:conclusion}.

\section{Preliminaries}
\label{sec:preliminaries}

In this section, we review the concept of boson sampling, the most established model of losses used in the quantum optics literature, and the trace distance and its properties, needed in the understanding of our main result in section \ref{sec:mainresult}.
In this work we will use the notation $\ket{\bar{n}}=|n_{1}\rangle \otimes |n_{2}\rangle \otimes \ldots \otimes |n_{M}\rangle$, for the $M$-mode Fock basis,
where $\bar{n}$ correspond to a vector of $M$ integers and $|n_{i}\rangle$
is the Fock state corresponding to $n_i$ photons in mode $i$. 

\subsection{The ideal boson sampling model}
\label{subsec:introBS}
The boson sampling proposal concerns the interference of a multi-photon input state $\ket{\bar{1}_N}=|1\rangle^{\otimes N} \otimes |0\rangle^{\otimes M-N}$, 
over an $M$-mode linear optics 
interferometer modeled as a linear transformation of the annihilation operators
\begin{equation}
 \hat{b}_i=\sum_jU_{ij}\hat{a}_j,
\end{equation}
or $\mathbf{\hat{b}}=U\mathbf{\hat{a}}$ for an equivalent compact notation.
We remark that the unitarity of $U$ guarantees the preservation of the total photon number and that $U$ is an $M\times M$ matrix acting on the creation operators, to which corresponds a homomorphism $\varphi(U)$ of dimension $\binom{N+M-1}{N}$ acting on the $M$-mode Fock space \cite{Aaronson2013}.
At the output of the interferometer we implement a measurement in the photon number basis on each mode, where the probability of obtaining an outcome $\bar{z}$
reads $|\bra{\bar{z}}\varphi(U)\ket{\bar{n}}|^2$. When $\bar{z}$ corresponds to a string of bits, the probability outcome
is connected to the permanent of a submatrix of $U$, 
a crucial tool in the hardness proof of boson sampling \cite{Aaronson2013}.

A necessary condition in the proof of the hardness of boson sampling is the fact the $U$ needs to be a Haar random unitary.
The proposal by Reck et al. \cite{Reck94} showed that any general linear-optics transformation can be achieved with a planar circuit
composed of two-mode gates, using $M(M-1)/2$ gates distributed over $2M-3$ layers and $M$ parallel modes. A recent improvement in
\cite{Clements2016} remarkably brought this result to depth $M+1$, which is one unit close to the lower-bound $M$ obtained from a simple counting argument based on the degrees of freedom of a unitary matrix.
In order to make our result as general as possible we will consider the depth of the circuit $D$ as an additional free parameter.

In the initial boson sampling proposal, the proof necessitates a polynomial relation between the number of photons and modes ($M=\text{O}(N^5\log^2N)$).
In this work we consider the generalized relation
\begin{equation}
 N=kM^\gamma,
 \label{eq:photonscaling}
\end{equation}
where $0<k<1$ and $0<\gamma\leq1$.
It is easy to see that $\gamma=1/2$ corresponds to the bosonic birthday paradox ratio \cite{Arkhipov2012}.
This ratio ensures that for input states $\ket{\bar{n}}$ composed of single photons,
the probability of two or more boson bunching at the output is negligible (on average over the set of Haar random unitaries and on the asymptotic limit of a large system).
The case $\gamma=1/6$ corresponds to $M=\text{O}(N^6)$, which guarantees the condition in the hardness proof 
in \cite{Aaronson2013} to hold. 

Finally, $\gamma=1$ corresponds to the regime where the density of photons (with $k$ satisfying $k\le 1$ for single photons at the input) remains constant while the size of the system increases, as opposed to the original boson sampling proposal where it decreases with the size of the system. 

\subsection{Modeling losses}
\label{subsec:modelloss}

\begin{figure*}[!t!]
\centering
  \includegraphics[width=.9\linewidth]{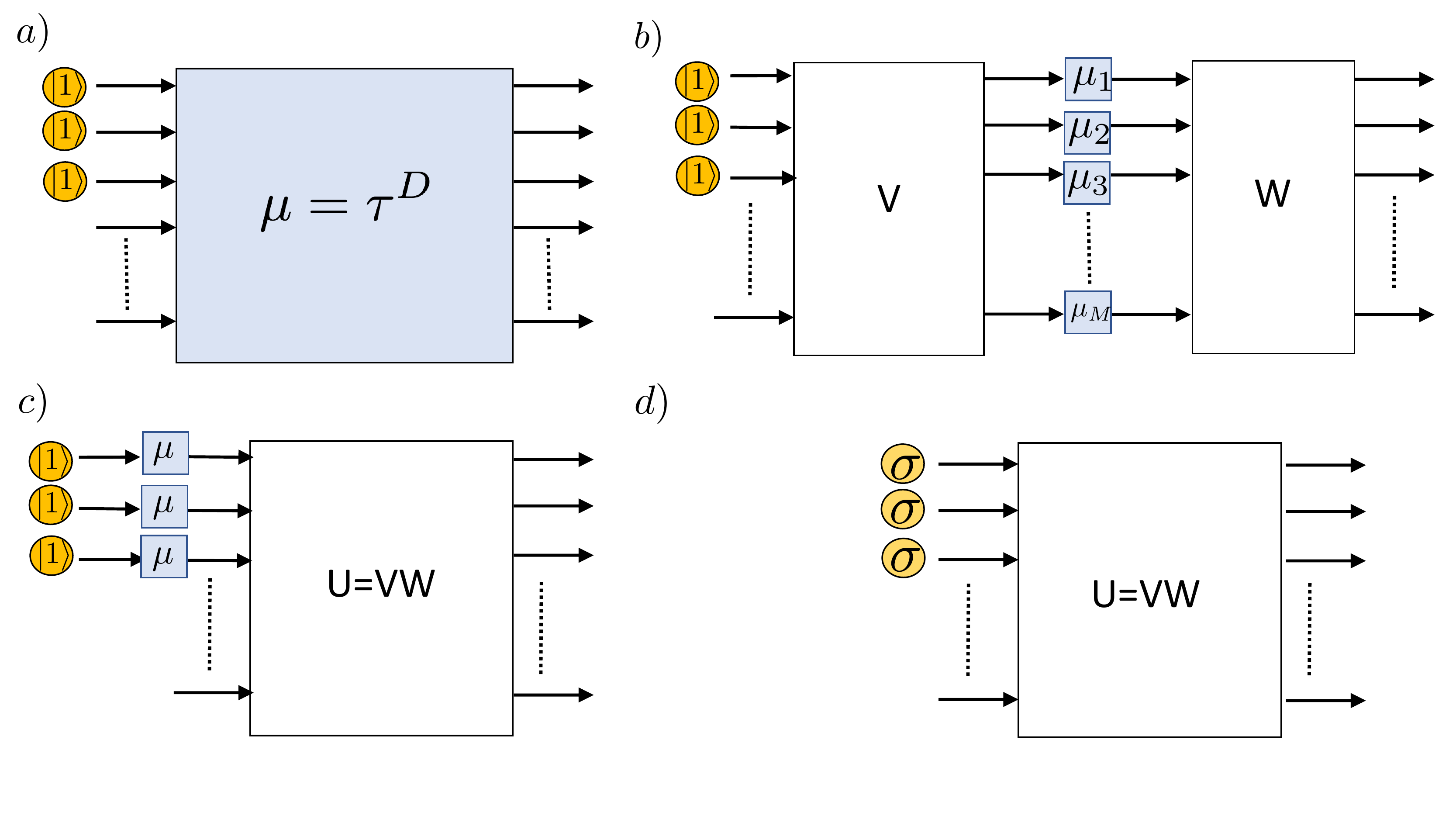}
\caption{a) $N$ single photons are sent over a lossy $M$ mode linear optical circuit composed of $D$
layers of two-mode couplers. An uniform transmission  per coupler of $\tau$  leads to an exponentially decay of the transmission $\mu=\tau^D$;
b) the real scheme in a) is indistinguishable from a circuit composed by a lossless linear optics transformation $V$, followed by $M$ parallel
set of pure-loss channels of transmission $\mu_i$ each, and a final lossless linear optics transformation $W$;
c) In the case where all the $\mu_i$ are equal, the virtual representation can be simplified to
a layer of $M$ identical pure-loss channels of transmission $\mu$ followed by a virtual lossless linear optics transformation $U=VW$;
d) The action of a pure-loss channel of transmission $\mu$ on a single-photon state $\ket{1}$ is
equivalent to an erasure channel of probability $\mu$, which outputs the mixed state $\sigma=(1-\mu)\proj{0}+\mu\proj{1}$.}
\label{fig:figure2}
\end{figure*}

A lossy linear optics circuit, as in Figure \ref{fig:figure2} a), can be
mathematically modeled by a complex matrix $A$ satisfying
$AA^\dagger\leq I$ and transforming the annihilation operators of $M$ input modes $\mathbf{\hat{a}}$ and $M$
environment modes $\mathbf{{\hat{e}}}$ as
\begin{equation}
 \mathbf{\hat{a}}_{\text{out}}=A\mathbf{\hat{a}}_{\text{in}}+\sqrt{I-AA^\dagger}\mathbf{\hat{e}}.
\end{equation}
$A$ has a singular value decomposition $A=V\hat{\mu} W$, where $V$ and $W$ are unitary matrices and
$\hat{\mu}={\rm diag}(\sqrt{\mu}_1,\sqrt{\mu}_2,...,\sqrt{\mu}_M)$ is a diagonal matrix of real values satisfying $\mu_i\in[0,1]$.
The singular value decomposition has a very natural interpretation, see Figure \ref{fig:figure2} b), which is that the real interferometer with losses
characterized by $A$ has an equivalent circuit composed by a lossless linear optics transformation $V$, followed by $M$ parallel
set of pure-loss channels of transmission $\mu_i$ each, and a final lossless linear optics transformation $W$.
A pure-loss channel is equivalent to a coupling interaction of transmission $\mu_i$ between our physical mode $i$ and an environmental mode, see \cite{Giovannetti2015} for details.
The matrix $A$ can be efficiently inferred using a simple tomographic technique that only needs two-mode interferences of classical laser light \cite{Rahimi13}.

In practice, a linear optics circuit is composed of a network of two-mode couplers and single-mode phase gates, where each layer of gates of an $M$-mode linear optics circuit is given by a direct product of local $2\times2$ linear transformations and complex scalars, resulting in a $M\times M$ complex banded matrix $A_i$ of bandwidth $1$. The total linear optics circuit transformation results from the multiplication of $D$ matrices $A_i$, i.e., $A=A_1A_2...A_D$.

All currently existing architecture proposals to implement a boson sampling experiment,
integrated photonic circuits, fiber-optical links and 3D-waveguides,
suffer from exponential decay of the transmission with the length of the circuit.
An intuitive explanation is that every photon has a constant probability of being lost per unit of length of the circuit or per layer of coupling gates.
For a planar circuit composed of $D$ layers of gates, where every gate has a transmission coefficient $\tau$,
we obtain that all the $\mu_i$ are equal and the transmission follows an exponential decay rule $\mu=\tau^D$.
Because $\hat{\mu}=\sqrt{\tau^D} I$ commutes with any matrix,
we can simplify the virtual representation of $A$ to
a layer of $M$ identical pure-loss channels of transmission $\mu$ followed by a virtual lossless linear-optics transformation $U=VW$ (see Figure \ref{fig:figure2} c)).

The action of a pure-loss channel of transmission $\mu$ into a single-photon state $\ket{1}$ is
equivalent to an erasure channel of probability $\mu$, resulting into a mixed state $\sigma=(1-\mu)\proj{0}+\mu\proj{1}$,
see Figure \ref{fig:figure2} d).
Therefore, boson sampling over a realistic interferometer with uniform losses is equivalent to an ideal boson sampler
over its virtual circuit $U=VW$ where we replace each of the $N$ single-photons of the input state, located in the first $N$ modes, by
the state $\sigma$, leading to a global input state 
\begin{eqnarray}
    \rho_{\text{\text{in}}}&=&\sigma^{\otimes N}\otimes\proj{0}^{\otimes(M-N)}, \nonumber\\
    \text{with} && \sigma=(1-\mu)\proj{0}+\mu\proj{1}.
    \label{eq:input}
\end{eqnarray}

\subsection{Trace distance and its properties}
\label{sub:tracedistance}
Before presenting our result we need to provide some definitions and properties of the trace distance \cite{Nielsen10}.
The trace distance between two quantum states $\rho$ and $\sigma$ reads (here and below $\|... \|$ stands for  the $1$-norm)
\begin{eqnarray}
D(\rho,\sigma)=\frac{1}{2}\mathrm{Tr}\left(\sqrt{\left( \rho-\sigma\right)^2}\right)=
\frac{1}{2}||\rho-\sigma||.
\end{eqnarray}
When both density matrices are diagonal in the same basis, i.e.,
$\rho=\sum_ip_i\proj{i}$ and $\rho=\sum_iq_i\proj{i}$, the trace distance is equivalent to the 
total variation distance of its corresponding eigenvalues
\begin{eqnarray}
D(\rho,\sigma)=\frac{1}{2}\|\bar{p}-\bar{q}\|=D(\bar{p},\bar{q}).
\end{eqnarray}
For simplicity, and following \cite{Nielsen10},  we will use the abuse of notation $D(\bar{p},\bar{q})$ for the total variation distance between two probability distributions. In what follows we will use three important properties of the trace distance: 
\begin{enumerate}
    \item Its invariance under unitary transformation 
    $D(U\rho U^\dagger,U\sigma U^\dagger)=D(\rho,\sigma)$.
    \item Its contractivity under a trace-preserving CP map $\mathcal{M}$, i.e.,
    $D(\mathcal{M}(\rho),\mathcal{M}(\sigma))\leq D(\rho,\sigma)$.
    \item Let $\{E_x\}$ be a POVM, with $p(x)=\text{Tr}[\rho E_x]$ and $q(x)=\text{Tr}[\sigma E_x]$ the probabilities of obtaining a measurement outcome labeled by $x$. Then 
    $D(p(x),q(x))\leq D(\rho,\sigma)$. Where there always exist a POVM that saturates the bound.
\end{enumerate}

Using the triangle inequality one can prove the following lemma.
\begin{lem}\label{lem:Ncopies}
The trace distance between $N$ copies of $\rho$ and $\sigma$ satisfies  the bound
\begin{equation}
 D(\rho^{\otimes N},\sigma^{\otimes N})\leq ND(\rho,\sigma).
 \label{eq:Ndistance}
\end{equation}
\end{lem}
For simplicity we show the proof for $n=2$, where we make use of the triangle inequality
and $\| \rho\otimes \sigma \|=\| \rho \| \|\sigma \|$:
\begin{align*}
\| \rho^{\otimes 2} - \sigma^{\otimes 2}\| &= \| \rho^{\otimes 2} -
\sigma \otimes \rho+\sigma \otimes \rho - \sigma^{\otimes 2} \| \\
&\leq \| \rho^{\otimes 2} - \sigma \otimes \rho \|+
\|\sigma \otimes \rho - \sigma^{\otimes 2}\| \\
&=  \| (\rho-\sigma) \otimes \rho \|+ \|\sigma
\otimes (\rho-\sigma)  \| \\
&= 2 \|\rho-\sigma\|
\end{align*}
Its generalization to $n>2$ is straightforward.

\section{Main result}
\label{sec:mainresult}
In this section we summarize the results of this manuscript, 
the detailed derivation of the lemmas \ref{lem:classicalalgo}, \ref{lem:TN} in the proofs are presented in sections \ref{sec:thermalstates} and \ref{sec:tensornetwork}.
To make our proofs more accessible, we first restrict to the scenario of uniform loss and later generalize the results to arbitrary circuits in Section \ref{sec:nonuniformlosses}, where we show that any result that holds for uniform $\mu$ can be generalized to $\mu_{max}=\max_i\mu_i$.

\subsection{Lossy boson sampling as thermal noise}
The key element of our proof is to approximate the input state $\rho_{\text{in}}$ in equation (\ref{eq:input}) with the state $\rho_{\text{T}}=\sigma_{\text{th}}^{\otimes N}\otimes\proj{0}^{\otimes(M-N)}$,
composed of $N$ thermal states in the first $N$ modes and the remaining $M-N$ input modes in a vacuum state. A thermal state is given by the Bose-Einstein distribution
\begin{equation}
\sigma_{\text{th}} = (1-\lambda)\sum_{x=0}^\infty \lambda^x \proj{x},
 \label{eq:thermal}
\end{equation}
where $\lambda = \frac{z}{1+z}$, with $z$ being the average number of photons.
This allows us to proof the following theorem.

\begin{thm}\label{th:losses}
For any multi-photon interferometer circuit of uniform transmission $\mu$ satisfying
\begin{equation}
 \mu\leq\sqrt{\frac{\epsilon}{N}},
 \label{eq:lossescond}
\end{equation}
we have $D\left(\rho_{\text{in}},\rho_{\text{T}}\right)\leq\epsilon$.
\end{thm}
\begin{proof}
One can use Lemma \ref{lem:Ncopies} to obtain the bound
$D(\rho_{\text{T}},\rho_{\text{in}})\leq ND(\sigma_{\text{th}},\sigma)$.
To have $D(\rho_{\text{T}},\rho_{\text{in}})\le \epsilon $ we require that $D(\sigma_{\text{th}},\sigma)\le \epsilon/N$.  A simple calculation gives
$D(\sigma_{\text{th}},\sigma)=\frac{1}{2}\left(\lambda^2+|\mu-\lambda|+|\lambda(1-\lambda)-\mu|\right)$. There are three cases: (i) $\mu \le \lambda(1-\lambda)$, (ii) $ \lambda(1-\lambda)\le \mu \le \lambda$ and (iii) $\lambda\le \mu$. All of them lead to the condition $\lambda\le \sqrt{\epsilon/N}$ and three intervals for   $\mu$  given  $\lambda$. The union of the latter reads
$\lambda- \epsilon/N\le \mu \le \lambda(1-\lambda) +\epsilon/N$.  We are looking for  any $\lambda$  and the maximal possible value of $\mu$ satisfying  $D(\sigma_{\text{th}},\sigma)\le \epsilon/N$.  Notice  that the upper bound $\mu\le \lambda(1-\lambda) +\epsilon/N $ grows with $\lambda$ (for $\lambda\le 1/2$) and  reaches its  maximal value for $\lambda_{\mathrm{max}} =  \sqrt{\epsilon/N}$ resulting in equation (\ref{eq:lossescond}).  Without loss of generality, we can set $\lambda =\mu$.
\end{proof}

In what follows we use the notation $p_{\text{in}}(\bar{n})$ and $p_{\text{T}}(\bar{n})$ for the outcome probabilities resulting from applying a linear interferometer to $\rho_{\text{in}}$ and $\rho_{\text{T}}$ followed by
a photon counting measurement. Because quantum operations and measurement can only decrease the trace distance,
a corollary of Theorem \ref{th:losses} is the bound 
$D\left(p_{\text{in}}(\bar{n}),p_{\text{T}}(\bar{n})\right)\leq\epsilon$. 
Therefore, any classical algorithm efficiently simulating
a boson sampling experiment with input state $\rho_{\text{T}}$ 
will be a good $\epsilon$-approximation of a lossy boson sampling experiment satisfying equation 
(\ref{eq:lossescond}). Theorem \ref{th:losses} formalizes, for multi-photon interference, the 
commonly held belief that any quantum supremacy experiment without access to error correction is equivalent to random noise after some noise threshold.

Let's analyze the prospect of a finite size boson sampling experiment under realistic conditions.
A simple calculation shows that for a multi-photon interference of $N$ photons
and transmission per coupler of $\tau=1-x$, the depth $\tilde{D}$ above which the output becomes 
$\epsilon$-close to a thermal state reads
\begin{equation}
    \tilde{D}=\frac{\log\frac{N}{\epsilon}}{2\log\frac{1}{1-x}}.
\end{equation}
Setting $\epsilon=10^{-6}$, a boson sampling experiment with $N=100$ photons, selected in order to discard a brute-force simulation via permanent calculations using \cite{Cliff17}, and realistic values for the loss of $x=10^{-3}$ per coupler in the circuit \cite{Taballione2018,Heck2014,Carolan2015,Harris2016}, we obtain $\tilde{D}=9.205$. This value is smaller than the
birthday paradox constraint $D=M\geq N^2$ and many orders of magnitude 
bellow the true boson sampling condition $O(N^5\log^2N)$. This is sufficient to
discard a quantum supremacy experiment within the boson sampling paradigm 
with current experimental capabilities and theoretical know-how. 

In order to satisfy the classical simulatability condition in equation~(\ref{eq:lossescond}) 
the transmission $\mu$ needs to decrease for fixed $\epsilon$ and increasing $N$.
This shows that the bound is relatively loose, and improvements are certainly possible. Nonetheless, it is sufficient to prove that platforms suffering from exponential decay of transmission, such as integrated photonics and fiber optics, can be efficiently simulatable on a classical computer, as shown in the next subsection.

\subsection{Exponential decaying transmission}
\label{subsec:BS expdecay}
In this section we provide an efficient algorithm for simulating multi-photon interference over
architectures with exponential decay of transmission.

\subsubsection{Sufficient condition for efficient classical simulation of multi-photon interference}
\label{subsec:classicalsimu}

The following lemma on the efficient classical simulation of photo-counting of interfering thermal states, which was implicit in  \cite{Rahimi2015}, is proven in section \ref{sec:thermalstates},
\begin{lem}\label{lem:classicalalgo}
There exists a polynomial time classical algorithm that simulates the
evolution of a thermal state $\rho_T$ over an ideal or lossy interferometer followed by measurement in the photon number basis,
where the output distribution is $\epsilon$-close to the ideal one and  
the computational running time scales as $O\left(\frac{MN^2}{\epsilon}\left[\log N+\log(1/\epsilon)\right]^2\right)$
for transmission $\mu$ satisfying equation (\ref{eq:lossescond}).
\end{lem}
As detailed in section \ref{sec:thermalstates} the algorithm combines the three following well-know facts in quantum optics.
Firstly, any thermal state $\rho_{\text{T}}$ has a Glauber-Sudarshan $P$-representation as a mixture of an $N$-mode tensor product of coherent states
$\ket{\bm{\alpha}}=\bigotimes_{i=1}^M \ket{\alpha_i}$ according to a Gaussian distribution $P(\bm{\alpha})$ \cite{Glauber}. 
Secondly, a linear-optical circuit characterized by a unitary matrix $U$ transforms a tensor product of coherent states $\bm{\alpha}$ into another tensor product of coherent states $\ket{\bm{\beta}}$ satisfying $\bm{\beta}=U\bm{\alpha}$.
Thirdly, coherent states follow a Poisson photon number distribution, which can be sampled efficiently.
Finally, we exploit the fact that there exist efficient constellations of coherent states that 
are arbitrarily close to the ideal thermal state with probability
distribution $P(\bm{\alpha})$ \cite{Lacerda17}. 
The combination of Theorem \ref{th:losses} and Lemma \ref{lem:classicalalgo} will allow us to simulate a lossy boson sampling architecture composed of $D$ layers of gates with exponentially decaying transmission $\mu=\tau^D$, as stated in the following theorem.

\begin{thm}\label{thm:classsimu}
The statistics $p_{\text{in}}(\bar{n})$ of $N$ photons interfering over an $M$-mode linear optics planar circuit of depth $D$, transmission $\tau$ per layer of
gates, and a relation between photons and modes given by  equation~(\ref{eq:photonscaling})
can be approximated with trace distance error $\epsilon$ in polynomial time for $D\geq D^*$, where \begin{equation}
 D^*=\frac{1}{2\log\left(\frac{1}{\tau}\right)}\left[\gamma\log M+
 \log\left(\frac{k}{\epsilon}\right)+\log2\right].
 \label{eq:circuitthreshold}
\end{equation}
\end{thm}

\begin{proof}
Any classical algorithm that generates a distribution $\tilde{p}_{\text{T}}(\bar{n})$ approximating the sampling from an ideal thermal state distribution $p_{\text{T}}(\bar{n})$,
where $\lambda=\mu$, satisfies the bound
\begin{widetext}
\begin{equation}
D(p_{\text{in}}(\bar{n}),\tilde{p}_{\text{T}}(\bar{n}))
\leq D(p_{\text{in}}(\bar{n}),p_{\text{T}}(\bar{n}))+D(\tilde{p}_{\text{T}}(\bar{n}),p_{\text{T}}(\bar{n})) 
\leq D(\rho_{\text{in}},\rho_{\text{T}})+D(\tilde{p}_{\text{T}}(\bar{n}),p_{\text{T}}(\bar{n})),
\label{eq:Dlimit}
\end{equation}
\end{widetext}
where we use the triangle inequality in the first inequality and the fact that a measurement over a quantum state can only decrease its trace norm in the second.
We can now apply Theorem \ref{th:losses} to set the bound $D(\rho_{\text{in}},\rho_\text{T})<\epsilon/2$ and Lemma \ref{lem:classicalalgo} to bound $D(\tilde{p}_{\text{T}}(\bar{n}),p_{\text{T}}(\bar{n}))<\epsilon/2$.
Then, the classical simulability condition $D\geq D^*$ can be trivially derived starting from the condition $\mu\leq\sqrt{\epsilon/(2N)}$ adapted from Theorem \ref{thm:classsimu}, replacing the transmission by $\mu=\tau^D$, taking the logarithm and replacing $N$ by eq.~(\ref{eq:photonscaling}).
Therefore, when the condition $D\geq D^*$ is satisfied, by properly selecting a thermal state $\rho_T$ satisfying $\lambda=\mu$ and running the algorithm sketched in the discussion after Lemma \ref{lem:classicalalgo}
(see  section \ref{sec:thermalstates} for details), $\tilde{p}_{\text{T}}(\bar{n})$ provides an $\epsilon$-approximation of $p_{\text{in}}(\bar{n})$ in polynomial time.
\end{proof}

\subsubsection{Simulation of shallow boson sampling circuits using tensor networks}

In section \ref{sec:tensornetwork} we show how one can simulate an ideal boson sampling circuit using tensor network techniques,
which can be summarized in the following lemma.

\begin{lem}\label{lem:TN}
An ideal boson sampling circuit with $N$ interfering photons over an $M$-mode linear interferometer of depth $D$
can be simulated exactly using tensor networks with a computational running time $\text{O}(M^2(N+1)^{8D})$.
\end{lem}

Tensor networks are a way of encoding quantum states and operating with them that have proven to be very successful
in many-body physics~\cite{Wh92, OeRo95, VePoCi04, Orus14}. Our tensor network proof is a quantum optics version,
adapted from \cite{Temme2012}, of the well-know result that logarithmic-depth planar circuit of $M$ qudits can be simulated on polynomial time
\cite{Jozsa06}. It is easy to see that when the depth of the circuit scales logarithmically with the number of modes $M$
and $N$ satisfies equation (\ref{eq:photonscaling}), our algorithm runs in quasipolynomial time.
This is due to the unbounded nature of the Hilbert space of optical modes.
In order to have an exact simulation we need to fix the local dimension on each mode to be as large as the total number of photons in the circuit, which results into a quasipolynomial-time algorithm.

\subsubsection{Efficient simulation of architectures with exponential decaying transmission}

Now, combining theorem \ref{thm:classsimu} and lemma \ref{lem:TN} we can classically simulate any multi-photon interference architecture that has an exponential decaying transmission, as stated by the following theorem.
\begin{thm}\label{th:expdecay}
The statistics of $N$ photons interfering over an $M$-mode planar linear-optical circuit of depth $D$, transmission $\tau$ per layer of
gates, and a relation between photons and modes given by  equation~(\ref{eq:photonscaling})
can be approximated with trace distance error $\epsilon$ in polynomial time for $D\geq D^*$ and in
quasi-polynomial time for $D\leq D^*$, where $D^*$ is defined in equation (\ref{eq:circuitthreshold}).
\end{thm}
\begin{proof}
The classical simulability under condition $D\geq D^*$ is a direct corollary of Theorem \ref{thm:classsimu}. For a uniform losses circuit satisfying $D\leq D^*$ we exploit the equivalence between a lossy multi-photon interference and an ideal interference over the circuit $U$ with input state $\rho_{\text{in}}=\sigma^N\otimes\proj{0}^{M-N}$, as explained in subsection \ref{subsec:modelloss}.
We model the initial state $\rho_{\text{in}}$ by starting with $N$ single-photons and 
keeping each one with probability $\mu=\tau^D$, while the rest are transformed to vacuum inputs.
We then proceed with the tensor network simulation of $U$ using lemma \ref{lem:TN}, where we just need to change the input tensor accordingly to the random sequence of (surviving) input single-photons. This leads to a quasi-polynomial time algorithm with a running time $\text{O}\left(M^{\frac{\gamma^2}{2\log\frac{1}{\epsilon}}\log M}\right)$.
\end{proof}

\section{Additional results}
\label{sec:extra}
In this section we extend the previous results to architectures suffering from algebraic decay of transmission and to other quantum optics alternatives to boson sampling.

\subsection{Algebraic decay of transmission}
\label{subsec:algdecay}

Not all optical architectures suffer from an exponential decay of the transmission, for example  free-space optics suffers from a
decay of transmission scaling as $1/D^2$. Suppose that a given architecture follows the following algebraic decay of losses
\begin{equation}
 \mu=\left(1+\frac{D}{d}\right)^{-\beta},
 \label{eq:algebraicdecay}
\end{equation}
where $d$ is a length scale that together with the parameter $\beta$ model the algebraic decrease of transmission.
Then theorem \ref{thm:classsimu} can be adapted to the following weaker form.
\begin{coro}\label{cor1}
The statistics of $N$ photons interfering over an $M$-mode linear optics planar circuit of depth $D$, with algebraic losses given by eq.~(\ref{eq:algebraicdecay}), and a relation between photons and modes given by eq.~(\ref{eq:photonscaling})
can be approximated with trace distance error $\epsilon$ in polynomial time when $D\geq D^*$, where
\begin{equation}
 D^*= d\left[\left(\frac{2k}{\epsilon}\right)^{\frac{1}{2\beta}}M^{\frac{\gamma}{2\beta}}-1\right].
 \label{eq:algcondition}
\end{equation}
\end{coro}

It is not difficult to check that when the condition $\gamma/\beta<2$ is satisfied, 
there always exist an $M^*$ such that
the condition $D^*\leq M-1$ is satisfied for all $M\geq M^*$. This shows that any boson sampling experiment, which needs a depth $D=M$, 
on an architecture causing algebraic decay of optical transmission satisfying $\gamma/\beta<2$ will be classically simulatable by an approximation by thermal noise sampling for all $M\geq M^*$.

\subsection{Generalization to alternative boson sampling proposals}
\label{subsec:BSalternatives}

Scattershot boson sampling was presented in \cite{Lund14} to circumvent the main
problem of the non-deterministic nature of state-of-the-art photon sources, where the probability of firing
$N$ photon at the same time decays exponentially. The protocol starts by generating $M$ two-mode squeezed vacuum
states,
\begin{equation}
 \ket{\psi}=\sqrt{1-\lambda}\sum_{n=0}^\infty \lambda^{n/2}\ket{n}\ket{n},
 \label{eq:tmsvs}
\end{equation}
where $\lambda$ is the same parameter as in the definition of a thermal state, 
as a two-mode squeezed vacuum state is its purification.
Then we send half of the two-mode squeezed vacuum states through a boson sampling circuit, while the remaining modes are used to herald
the presence of photons. By properly tuning the squeezing parameter $\lambda$ one can guarantee that most of the heralded
sequences are collision-free, i.e., satisfy the birthday-paradox condition. The price to pay is that the modes where the photons
enter the circuit are completely randomized, which is not a problem for boson sampling as the circuit is anyway randomly generated according to the Haar random distribution.
Because right after the heralding process the setup is strictly equivalent to a traditional boson sampling device,
up to the randomization of the modes where the single-photons enter the interferometer, both of our simulation algorithms
(thermal state sampling and tensor network simulation) can be trivially adapted.
We only need to randomly generate
valid heralding sequences following the distribution given by eq.~(\ref{eq:tmsvs}) and depending on the obtained heralded value
we run the boson sampling simulations presented in subsection \ref{subsec:classicalsimu} and
detailed in section \ref{sec:thermalstates}. The only difference is that now the
input photons enter the interferometer on a random selection of $N$ input modes.

More recently, a variant of boson sampling, where photon detectors are replaced by a Gaussian measurement, has been proposed \cite{Lund17,Chak17}.
Because quantum operations and any measurement can only decrease the trace distance, the outcome statistics 
of this alternative proposals will also remain $\epsilon$-close. The evolved thermal state being Gaussian,
we can extend our result to this scenario by using well-know techniques of simulating Gaussian measurement over Gaussian states \cite{Veit12,Mari12}.

\section{Proof of Lemma \ref{lem:classicalalgo}}
\label{sec:thermalstates}

An idealized algorithm for simulating the photo-counting of a set of interfering thermal states is composed of three steps.
Firstly, any thermal states $\rho_{\text{T}}$ has a Glauber-Sudarshan $P$-representation as a mixture of $N$-mode coherent states
$\ket{\bm{\alpha}}\equiv\ket{\alpha_1, ..., \alpha_N}=\bigotimes_{i=1}^N \ket{\alpha_i}$ according to a Gaussian distribution
\begin{equation}
\rho_{\text{T}}=\int_{\mathbb{C}^N}d\bm{\alpha}p(\bm{\alpha})\proj{\bm{\alpha}}\otimes\proj{\mathbf{0}}^{M-N},
\label{eq:Pfunction}
\end{equation}
where
\begin{equation}
p(\bm{\alpha})=\prod_{i=1}^N\left[\frac{d^2\alpha_i }{\pi z}\,
\exp\left(-\frac{|\alpha_i|^2}{z}\right)\right].
\end{equation}
Secondly, a linear-optical circuit characterized by a unitary matrix $U$ transforms a tensor product of coherent states
$\bigotimes_{i=1}^M \ket{\alpha_i}$,
into another tensor product of coherent states $\bigotimes_{i=1}^M \ket{\beta_i}$ with amplitudes
\begin{equation}\label{5}
\beta_i=\sum_{j=1}^M U_{ij} \alpha_j.
\end{equation}
In other words, coherent states remain in a tensor product form while evolved through a linear optical circuit.
Thirdly, coherent states follow a Poisson photon number distribution
\begin{equation}
 P(n_i,\beta_i)=e^{-|\beta_i|^2}\frac{|\beta_i|^{2n_i}}{n_i!}.
 \label{eq:Poisson}
\end{equation}
Therefore, a concatenation of three stochastic processes 
simulates the photo-counting of a set of interfering thermal states. 
The first process generates a complex vector $\bm{\alpha}$ following the probability distribution $p(\bm{\alpha})$. The second one  applies
the map $\mathcal{U}$ to the vector $\bm{\alpha}$ generating the output $\bm{\beta}=U\bm{\alpha}$. The third process
 $\mathcal{P}$ generates an $M$-dimensional vector $\bar{n}$ from $\bm{\beta}$ by sampling from
the $M$-dimensional Poisson distribution  $p(\bar{n},\bm{\beta})$, where averaging over $\bm{\beta}$ gives
\begin{equation}
   p_{\text{T}}(\bar{n})=\int d\bm{\beta}p(\bm{\beta})p(\bar{n},\bm{\beta})
\end{equation}
This algorithm is an idealized one, as it  assumes access to oracles that sample exactly from Gaussian and Poisson distributions.

In order to build a realistic algorithm we define a new three step process, where the sampling oracles are replaced by efficient approximation algorithms.
The first step, detailed in subsection \ref{subsec:constellation}, consist of sampling from a $N$-mode constellation $\mathcal{C}_{N,\lambda, m^2}(\bm{\alpha})$ 
composed of $m^2N$ coherent states of distribution $\tilde{p}(\bm{\alpha})$, satisfying
\begin{equation}
    \sum_{i=1}^{m^2N}\tilde{p}(\bm{\alpha}_i)\proj{\bm{\alpha}_i}=\tilde{\rho}_{\text{T}},
    \label{eq:constellation}
\end{equation}
which efficiently approximates the $N$-mode thermal state $\rho_{\text{T}}$ \cite{Lacerda17}.
The $N$-mode constellation is composed of 
$N$ identical single-mode constellations  $\mathcal{C}_{\lambda, m^2}(\alpha)$ of size $m^2$
providing each a good approximation of the single-mode thermal state $\sigma_{\text{th}}$.
The second step $\mathcal{\tilde{U}}$ implements an approximation of matrix multiplication,
discussed in subsection \ref{subsec:matrixmult}, mapping $\bm{\alpha}$ to $\bm{\beta}$,
transforming the constellation $\mathcal{C}_{N,\lambda, m^2}(\bm{\alpha})$ into $\mathcal{C}_{N,\lambda, m^2}(\bm{\beta})$.
Finally the third step $\mathcal{\tilde{P}}$ generates $\tilde{p}(\bar{n},\bm{\beta})$, an
approximation of an $M$-dimensional Poisson distributions (\ref{eq:Poisson}) satisfying
\begin{equation}
    \sum_j^{m^2 N}\tilde{p}(\bar{n},\bm{\beta}_j)=\tilde{p}_{T}(\bar{n}).
\end{equation}
To approximate the Poisson distribution we use a scalable number of Bernouilli trials, 
as detailed in subsection \ref{subsec:Poisson} following \cite{Barb84}. 

\subsection{Error analysis}
We want to show that the trace distance between the ideal and approximate algorithms, above described,  satisfies $D(p_{\text{T}}(\bar{n}),\tilde{p}_{\text{T}}(\bar{n})\leq \epsilon$ while the algorithm 
remaining polynomial-time. It is easy to see from the definition of 
$\mathcal{P}$, $\mathcal{\tilde{P}}$, $\mathcal{U}$, $\mathcal{\tilde{U}}$ that
\begin{equation}
     \|p_{\text{T}}(\bar{n})-\tilde{p}_{\text{T}}(\bar{n}) \|=
     \|\mathcal{P}\circ\mathcal{U}\left(\rho_{\text{T}}\right)-
    \mathcal{\tilde{P}}\circ\mathcal{\tilde{U}}\left(\tilde{\rho}_{\text{T}}\right) \|.
\end{equation}

One can rewrite it as the following norm of a linear combination of terms,
\begin{widetext}
 \begin{equation}
 \|p_{\text{T}}(\bar{n})-\tilde{p}_{\text{T}}(\bar{n})\|=
   \|\mathcal{P}\circ\mathcal{U}\left(\rho_{\text{T}}\right)-\mathcal{P}\circ\mathcal{U}\left(\tilde{\rho}_{\text{T}}\right)
  +\mathcal{P}\circ\mathcal{U}\left(\tilde{\rho}_{\text{T}}\right)-\mathcal{P}\circ\mathcal{\tilde{U}}\left(\tilde{\rho}_{\text{T}}\right)
  +\mathcal{P}\circ\mathcal{\tilde{U}}\left(\tilde{\rho}_{\text{T}}\right)-\mathcal{\tilde{P}}\circ\mathcal{\tilde{U}}\left(\tilde{\rho}_{\text{T}}\right) \|
\end{equation}
which using the triangle inequality and the fact that a trace preserving map, such as $\mathcal{U},\mathcal{P},\tilde{\mathcal{U}},\tilde{\mathcal{P}}$ can only decrease the trace distance, we obtain the upper-bound
\begin{equation}
  \|p_{\text{T}}(\bar{n})-\tilde{p}_{\text{T}}(\bar{n})\|\leq
   \|\rho_{\text{T}}-\tilde{\rho}_{\text{T}} \|+
   \|\mathcal{U}\left(\tilde{\rho}_{\text{T}}\right)-\mathcal{\tilde{U}}\left(\tilde{\rho}_{\text{T}}\right) \|
  + \|\mathcal{P}\circ\mathcal{\tilde{U}}\left(\tilde{\rho}_{\text{T}}\right)-\mathcal{\tilde{P}}\circ\mathcal{\tilde{U}}\left(\tilde{\rho}_{\text{T}}\right) \|,
\end{equation}
that can be simplified using the definition of the constellation $\mathcal{C}_{N,\lambda,m}(\bm{\alpha})$ and $\mathcal{C}_{N,\lambda,m}(\bm{\beta})$, resulting in
\begin{equation}
 \|p_{\text{T}}(\bar{n})-\tilde{p}_{\text{T}}(\bar{n})\|\leq
  \|\rho_{\text{T}}-\tilde{\rho}_{\text{T}}\|+
  \max_{\mathcal{C}(\bm{\alpha})}\|\mathcal{U}\left(\proj{\bm{\alpha}}\right)-\mathcal{\tilde{U}}\left(\proj{\bm{\alpha}}\right)\|
  +\max_{\mathcal{C}(\bm{\beta})}\|\mathcal{P}\left(\proj{\bm{\beta}}\right)-\mathcal{\tilde{P}}\left(\proj{\bm{\beta}}\right)\|,\label{eq:errorbound}
\end{equation}
\end{widetext}
where $\mathcal{C}(\bm{\beta})=\mathcal{U}\left(\mathcal{C}(\bm{\alpha})\right)$.
In what follows we will show how one can build efficient algorithm for each step such that the last three terms on the r.h.s. of (\ref{eq:errorbound}) are bounded by $2\epsilon/3$, leading to 
$D(p_{\text{T}}(\bar{n}),\tilde{p}_{\text{T}}(\bar{n}))\leq\epsilon$.

\subsection{Efficient coherent states constellations}
\label{subsec:constellation}
We follow the derivation in \cite{Lacerda17} to bound the trace norm between a single-mode thermal state $\sigma_{\text{th}}$ and the single-mode coherent state constellation 
\begin{equation}
    \mathcal{C}_{\lambda,m}(\alpha)=\{P_{\lambda,m}(\alpha),\proj{\alpha}\}_{\alpha\in\mathbb{C}},
\end{equation}
which needs to have at least the same first and second moments. 
The constellation is described by the $P$-distribution $P_{\lambda,m}(\alpha)$
supported on $m^2$ points, such that 
\begin{equation}
    \frac{V}{2}P_{\lambda,m}
    \left[\sqrt{\frac{V}{2}}\left(x+iy\right)\right]
    =P_{X_m}(x)P_{X_m}(y),
\end{equation}
where $P_{X_m}(x)$ is one constellations  of the normal distribution $X\sim N(0,1)$ \cite{Wu2010}.
In terms of random variables, $\alpha_m=\sqrt{\frac{V}{2}}\left(X_m+i X_m'\right)$,
where $X_m$ and $X_m'$ are independent realizations of the given constellation. The factor
$\sqrt{V}$ ensures that the resulting $P$ function has the same variance
\begin{equation}
    V=\frac{\lambda}{1-\lambda}
    \label{eq:var}
\end{equation}
as the one corresponding to the thermal state $\sigma_{\text{th}}$, where the factor $1/\sqrt{2}$ takes care of the conversion from two real variables to one complex variable.
As shown in \cite{Temme10} the $\chi^2$-divergence bounds the trace norm as
\begin{equation}
    ||\rho-\sigma||^2\leq \chi^2(\rho,\sigma)=\text{Tr}\left[\left(\rho\sigma^{-1/2}\right)^2\right].
\end{equation}
As proven in the Appendix of \cite{Lacerda17}, the quantum $\chi^2$-divergence
$\chi^2(\sigma_{\text{th}},\mathcal{C}_{\lambda,m}(\alpha))$ satisfies the upper-bound 
\begin{equation}
    1+\chi^2(\sigma_{\text{th}},\mathcal{C}_{\lambda,m}(\alpha))=(1+\chi^2(P_{X_m},P_{X}))^2
\end{equation}

The Gauss-Hermite constellation $X_m$ of size $m$ of the normal distribution $X\sim N(0,1)$ is given by the $m$ roots of the $m^{\text{th}}$ Hermite polynomial, with weight $P_{X_m}$ selected to provide exact integration with respect to $P_{X_m}$ for all polynomials up to degree $2m-1$ \cite{Wu2010}. Then one can relate the $\chi^2$-divergence between the normal distribution and the Gauss-Hermite constellation $\chi^2(P_{X_m},P_{X})$ to the moments of the Hermite polynomials of $X_m$ \cite{Lacerda17,Wu2010}
\begin{equation}
    1+\chi^2(P_{X_m},P_{X})=\sum_{k=0}^\infty\frac{1}{k!}\left(\frac{s}{1+s}\right)^k
    |\mathbb{E}\left[H_k\left(X_m\right)\right]|^2
\end{equation}
where 
\begin{equation}
    s=\frac{V}{\sqrt{V(V+1)}-V}.
    \label{eq:s}
\end{equation}
Because $\mathbb{E}\left[H_k\left(X_m\right)\right]=0$ for odd $k$ and by definition of the Gauss-Hermite quadrature $\mathbb{E}\left[H_k\left(X_m\right)\right]=0$ for all $k\leq 2m-1$,
together with using equation (\ref{eq:s}) and definition (\ref{eq:var}) to obtain the equality 
\begin{equation}
  \frac{s}{1+s}=\sqrt{\lambda},
\end{equation}
we reach the simplification
\begin{equation}
    \chi^2(P_{X_m},P_{X})=\sum_{k\geq m}^\infty\frac{\lambda^k}{(2k)!}
   |\mathbb{E}\left[H_{2k}\left(X_m\right)\right]|^2.
\end{equation}
Finally, following the proof of Theorem 8 in \cite{Wu2010} we can show 
\begin{equation}
    \chi^2(P_{X_m},P_{X})\leq 2\kappa^2\sum_{k\geq m}^\infty\lambda^k=2\kappa^2\frac{\lambda^m}{1-\lambda},
\end{equation}
where $\kappa$ is a constant such that $2\kappa^2\approx2.36$.
This provides us with an upper-bound to the trace distance between the single-mode constellation 
and the single-mode thermal state
\begin{eqnarray}
    D\left(\sigma_{\text{th}},\mathcal{C}_{\lambda,m}(\alpha)\right)&\leq&
   \frac{1}{2} \left(\left[1+2\kappa^2\frac{\lambda^m}{1-\lambda}\right]^2- 1\right) \nonumber\\
    &&\leq 3\kappa^2\frac{\lambda^m}{1-\lambda}.
\end{eqnarray}
Using Lemma \ref{lem:Ncopies} one can trivially obtain an upper-bound between the $N$-mode constellation and thermal state from its single-mode version
\begin{equation}
    D\left(\rho_{\text{T}},\mathcal{C}_{N,\lambda,m}(\bm{\alpha})\right)\leq 3\kappa^2N\frac{\lambda^m}{1-\lambda}.
    \label{eq:constsize}
\end{equation}

\subsection{Matrix multiplication}
\label{subsec:matrixmult}
The transformation of $\bm{\beta}= U\bm{\alpha}$ can be approximated using standard numerical linear algebra within error $\epsilon$ in $\text{O}(MN)$ operations \cite{Knuth}. It is then easy, using the fidelity upper-bound of the trace norm, the well-know exponential decrease of the overlap between two coherent states, and the inequality $1-e^{-x}\leq x$ to derive the bound
\begin{widetext}
\begin{equation}
||\mathcal{U}\left(\proj{\bm{\alpha}}\right)-\mathcal{\tilde{U}}\left(\proj{\bm{\alpha}}\right)||
\leq 2\sqrt{1-F\left(\mathcal{U}\left(\proj{\bm{\alpha}}\right),\mathcal{\tilde{U}}\left(\proj{\bm{\alpha}}\right)\right)}
\leq 2\sqrt{1-e^{-|\bm{\tilde{\beta}}-\bm{\beta}|^2}}\leq2|\bm{\tilde{\beta}}-\bm{\beta}|\leq 2\epsilon,
\end{equation}
\end{widetext}
which shows that the error on the coherent sates is upper bounded by the error of the matrix multiplication, which can be made arbitrary small with a polynomial-time overhead. 

\subsection{Approximate Poisson sampling with Bernouilli trials}
\label{subsec:Poisson}
Starting from an $M$-dimensional vector $\bm{\beta}$  the map $\mathcal{P}(\bm{\beta})$  
outputs samples $\bar{n}$ from $M$ Poisson distributions, i.e., 
for each $\beta_l$  with $l = 1, \ldots, M$ 
the Poisson distributions $P(n_l,|\beta_l|^2)$  (\ref{eq:Poisson}) is sampled
via independent Bernoulli trials.
To determine the number of Bernoulli trials $t$ to achieve a given error we can use the trace distance bound \cite{Barb84}
\begin{widetext}
\begin{equation}
 \frac{1}{2} \sum_{k=0}^\infty\Bigl| \binom{t}{k}p^k_B (1-p_B)^{t-k} - P(k,x)\Bigr| \le (1-e^{-x})\frac{x}{t}\leq\frac{x^2}{t}
 \label{eq:28}
\end{equation}
\end{widetext}
between the probability distribution of the  sum of $t$ independent Bernoulli trials, $S_t = \xi_1+ \ldots  + \xi_t$, with  $p_B(\xi = 1) = x/t$, $p_B(\xi = 0) = 1-x/t$
and the  Poisson distribution $P(n,x)$. This implies that the error of the $M$-dimensional output $\bar{n}$ 
will be bounded as
\begin{equation}
D\left(\mathcal{P}\left(\proj{\bm{\beta}}\right),\mathcal{\tilde{P}}\left(\proj{\bm{\beta}}\right)\right)
\leq \frac{M}{t}\max_{\bm{\beta}}|\bm{\beta}|^4.
\end{equation}

Because the unitary matrix $U$ preserves the norm of a vector, it is easy to show
\begin{equation}
    \max_{\bm{\beta}}|\bm{\beta}|\leq N\max|\alpha|\leq N\sqrt{2m},
\end{equation}
where we used the fact that $\bm{\alpha}$ results from $N$identical  single-mode constellations of coherent states and Krasikov's upper bound on the roots of Hermite polynomials \cite{Krasikov01}, which define the location of the coherent states of the constellation. Therefore, the number of necessary Bernoulli trials $t$ for simulation of a Poisson distribution with a trace distance error $2\epsilon/3$ reads
\begin{equation}
    t=6\frac{MN^2}{\epsilon}m^2.
    \label{eq:bernouilli}
\end{equation}

\subsection{Algorithm scaling}
The algorithm is composed of three steps. First, the random generation of a $N$-mode coherent state,
a task with a computational cost scaling as $O(N)$. Secondly, the matrix multiplication $\bm{\beta}=U\bm{\alpha}$, with a computational cost scaling as $O(MN)$. Finally, we need to approximate a Poisson distribution by sampling from $t$ Bernoulli trials, 
as specified in equation (\ref{eq:bernouilli}). Using equation (\ref{eq:constsize}) and 
that  $\lambda=\mu$ we obtain
\begin{equation}
    m=\frac{\log N+\log(1/\epsilon)+\log(1/(1-\mu))+2\log(3\kappa)}{\log(1/\mu)}.
    \label{eq:m-scaling}
\end{equation}
It is easy to see that a constant $\mu$ guarantees
a computational cost scaling as $O\left(\frac{MN^2}{\epsilon}\left[\log N+\log(1/\epsilon)\right]^2\right)$ in terms of the number of input photons $N$, the transmission $M$ and the quality of the approximation $\epsilon$. Observe that  by theorem 2, $\mu \le \sqrt{\epsilon/N}$, hence the right hand side on equation (\ref{eq:m-scaling}) is  regular (for $\epsilon/N \le 1-\delta$ for all $\delta >0$). 


\section{Proof of Lemma \ref{lem:TN}}
\label{sec:tensornetwork}

A quantum state of $M$ bosonic modes with at most $N$ total number of photons reads
\begin{equation}
\ket{\psi}=\sum_{\{\bar{n}:|\bar{n}|=N\}} C_{n_1,n_2,...,n_{M}}\ket{n_1,n_2,...,n_{M}}.
\end{equation}

The memory and computational cost of a brute-force simulation is associated with the number of degrees of freedom of the coefficient $C_{n_1,n_2,...,n_{M}}$, which correspond to the size of the Hilbert space. In our case it is given by the binomial $\binom{N+M-1}{N}$ that grows exponentially if both $N$ and $M$ increase proportionally to each other.
The idea behind tensor networks is to interpret $C_{n_1,n_2,...,n_{M}}$ as a big tensor with $M$ free indices. A tensor that can be recovered from the contraction of a network of tensors of smaller size, where the virtual degrees of freedom contract each other leaving $M$ free parameters corresponding to the physical indices $n_i$. This provides a very intuitive representation of quantum states and allows for a very efficient encoding and manipulation when the states have a high degree of locality~\cite{Wh92, OeRo95, VePoCi04, Orus14}.

\subsection{Matrix product states}

In our case we are interested in the evolution of a particular example of a tensor network called matrix product states,
\begin{equation}
\ket{\psi}=\sum_{n_1=0}^{d_1}\sum_{n_2=0}^{d_2}...\sum_{n_{M}=0}^{d_M}B_{n_1}^{[1]}B_{n_2}^{[2]}...B_{n_M}^{[M]}\ket{n_1,...,n_{M}},
\label{eq:MPS}
\end{equation}
where $B^{[1]}_{n_{1}}$ is a transposed vector of dimension $\tilde{\chi}_1$, $B^{[M]}_{n_{M}}$ is a vector of dimension $\tilde{\chi}_M$, and $B^{[i]}_{n_{i}}$ for $1 < i < M$ is a matrix of dimension $\tilde{\chi}_i\times\tilde{\chi}_{i+1}$. The physical indexes $n_i$ take values
$\{0,1,...,d\}$.

\begin{figure*}[!t!]
\centering
\includegraphics[width=0.9\linewidth]{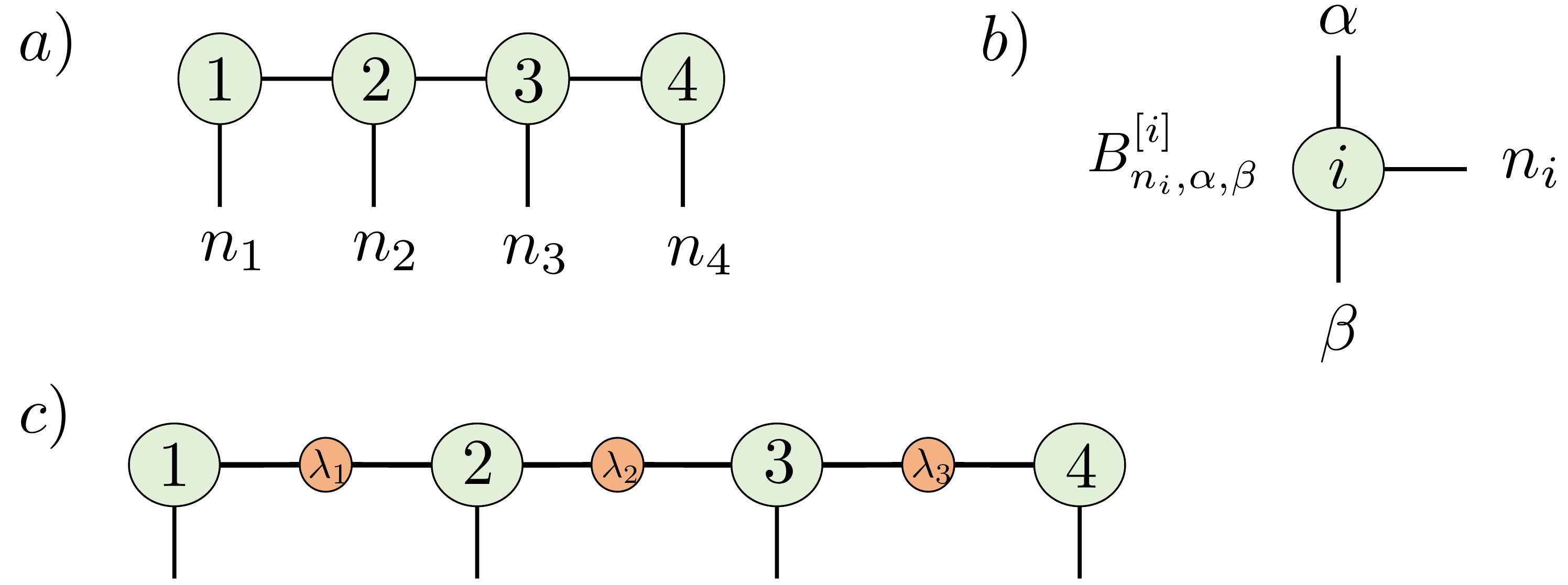}
\caption{\label{figure3}
(a) Graphical representation of a quantum state $|\psi\rangle$ corresponding to a matrix product state of four tensors, as defined in Eq.~\eqref{eq:MPS}.
(b) $B^{[i]}_{n_{i}, \alpha, \beta}$ is a tensor of rank three: represented by a vertex $i$ with three edges, one corresponding to the physical index $n_{i}$ and the two other to the virtual indexes
$\alpha$, and $\beta$.
(c) All matrix product states can be transformed into a canonical form where a diagonal matrix $\lambda$ of Schmidt coefficients is assigned to every edge between two vertices of (a).}
\end{figure*}
As shown in Figure \ref{figure3} (a),
one can associate to each matrix product state a 1-dimensional graph where each vertex is associated to a three index tensor
$B_{n_i,\alpha,\beta}^{[i]}$ (Figure \ref{figure3} (b)) and the edges determine the contraction rule of the tensor indices.

\subsubsection{Canonical form}
It is well known that any bipartite quantum state can be rewritten as
\begin{equation}
\ket{\psi}=\sum_{i=0}^{d_1}\sum_{j=0}^{d_2}c_{ij}\ket{ij}=\sum_{\alpha=0}^{\min{d_1,d_2}}\lambda_{\alpha}\ket{\varphi_\alpha}\ket{\psi_\alpha},
\end{equation}
where the Schmidt coefficients $\lambda_{\alpha}$ result from the singular value decomposition $c_{ij}=\sum_{\alpha}U_{i,\alpha}\lambda_{\alpha}\bar{V}_{j,\alpha}$.
Every matrix product state can be also transformed into a canonical form
\begin{widetext}
\begin{equation}\label{eq:canonical}
\ket{\psi}=\sum_{n_{1}, n_{2}, \ldots, n_{M}}
\left( \Gamma^{[1]}_{n_{1}} \lambda^{[1]} \Gamma^{[2]}_{n_{2}} \lambda^{[2]} \ldots \Gamma^{[N]}_{n_{M}} \right)
| n_{1}, n_{2}, \ldots, n_{M} \rangle
\end{equation}
\end{widetext}
by iteratively applying the singular value decomposition~\cite{Vi03, PeVeWoCi07}, with its  graph representation shown in Figure \ref{figure3} (c).
The matrices $\lambda^{[i]}$ are diagonal and contain the Schmidt coefficients corresponding to the bipartition of modes $(1,...,i)$ versus $(i+1,..,M)$. The Schmidt rank of each link reads $\chi_k$, and $\chi=\max_k\{\chi_k\}$ is called the bond dimension.
The total number of parameters and thus the storage cost for such a matrix product state scales as $\text{O}(M(d+1)\chi^{2})$.

\subsection{Simulating ideal linear optics circuits}

A linear optics circuit is composed of one-mode phase gates and two-mode couplers implementing an interaction between two adjacent modes.
In what follows we first explain how to update a matrix product state that goes under the evolution of linear optics gates and later discuss how to sample from the final output state.

\subsubsection{One mode gates (phase rotation)}
\label{subsec:1mode}

\begin{figure*}[!t!]
\centering
  \includegraphics[width=.9\linewidth]{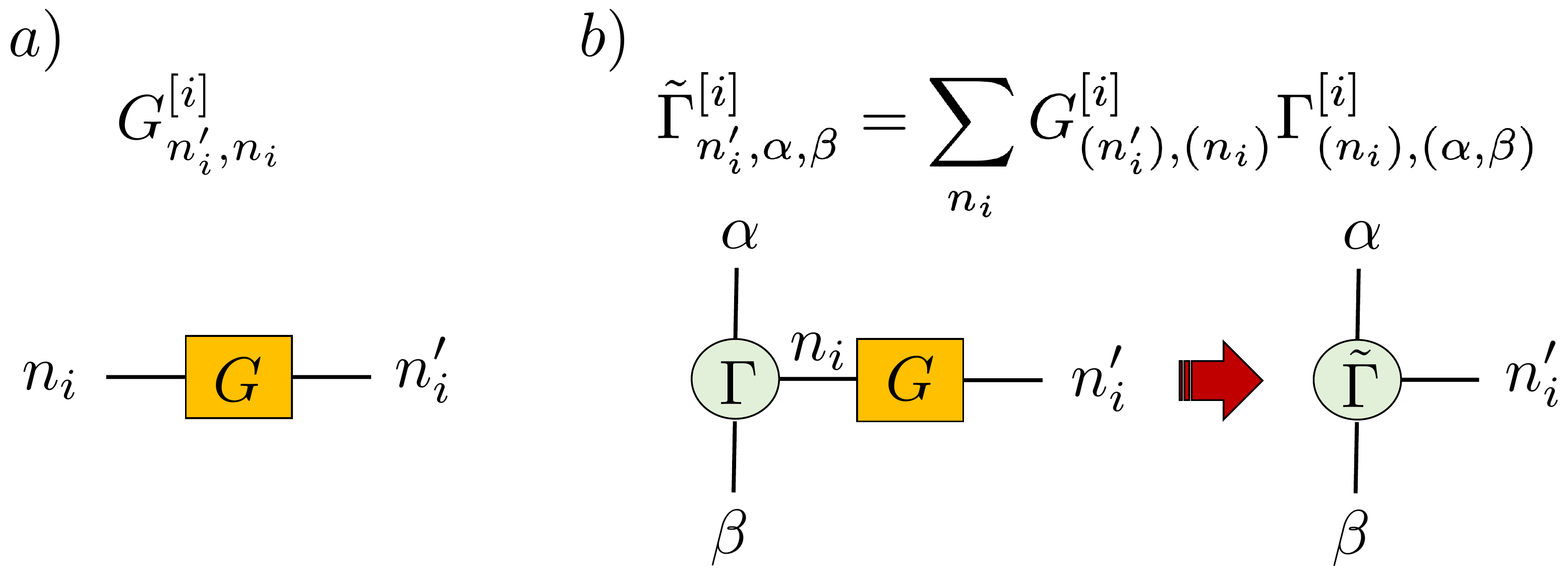}
\caption{a)A single-mode phase-rotation acting on mode $i$ is modeled by a matrix (tensor) $G^{[i]}_{n'_i,n_i}$ that transform the input physical indexes $n_i$
to the output physical indexes $n'_i$; b) The tensor $\Gamma^{[i]}$ of virtual indexes $\alpha,\beta$ and physical index $n_i$ is transformed
to the tensor $\tilde{\Gamma}^{[i]}$ by implementing a tensor contraction between the tensor $\Gamma^{[i]}$ and the phase-rotation $G^{[i]}$}.
\label{fig:figure4}
\end{figure*}

As shown in Figure \ref{fig:figure4}, a single-mode gate acting on mode $i$ is modeled by a matrix $G^{[i]}_{n'_i,n_i}$
that transform the input physical indices $n_i$ to the output physical indices $n'_i$.
The evolution corresponds to the contraction of the physical indices $n_i$ of $\Gamma^{[i]}_{\alpha,n_i,\beta}$
and $G^{[i]}_{n'_i,n_i}$ as
\begin{equation}
\tilde{\Gamma}^{'i}_{\alpha,\beta,n'_i}=\sum_{n_i}G^{[i]}_{(n'_i),(n_i)}\Gamma^{[i]}_{(n_i),(\alpha,\beta)}.
\end{equation}
A phase rotation $\theta$ has a matrix $G^{[i]}$ that is diagonal with coefficients $G^{[i]}=\exp(i\theta n_i)$. Therefore, the computational cost of the update of a single local gate scales as $\text{O}\left((d+1)\chi^2\right)$.
Notice that applying a single-mode gate will not change the Schmidt coefficients of the matrix product state, as it acts only on the physical indexes of one vertex of the graph.

\subsubsection{Two-mode couplers}
\label{subsec:2mode}

\begin{figure*}[!t!]
\centering
  \includegraphics[width=.9\linewidth]{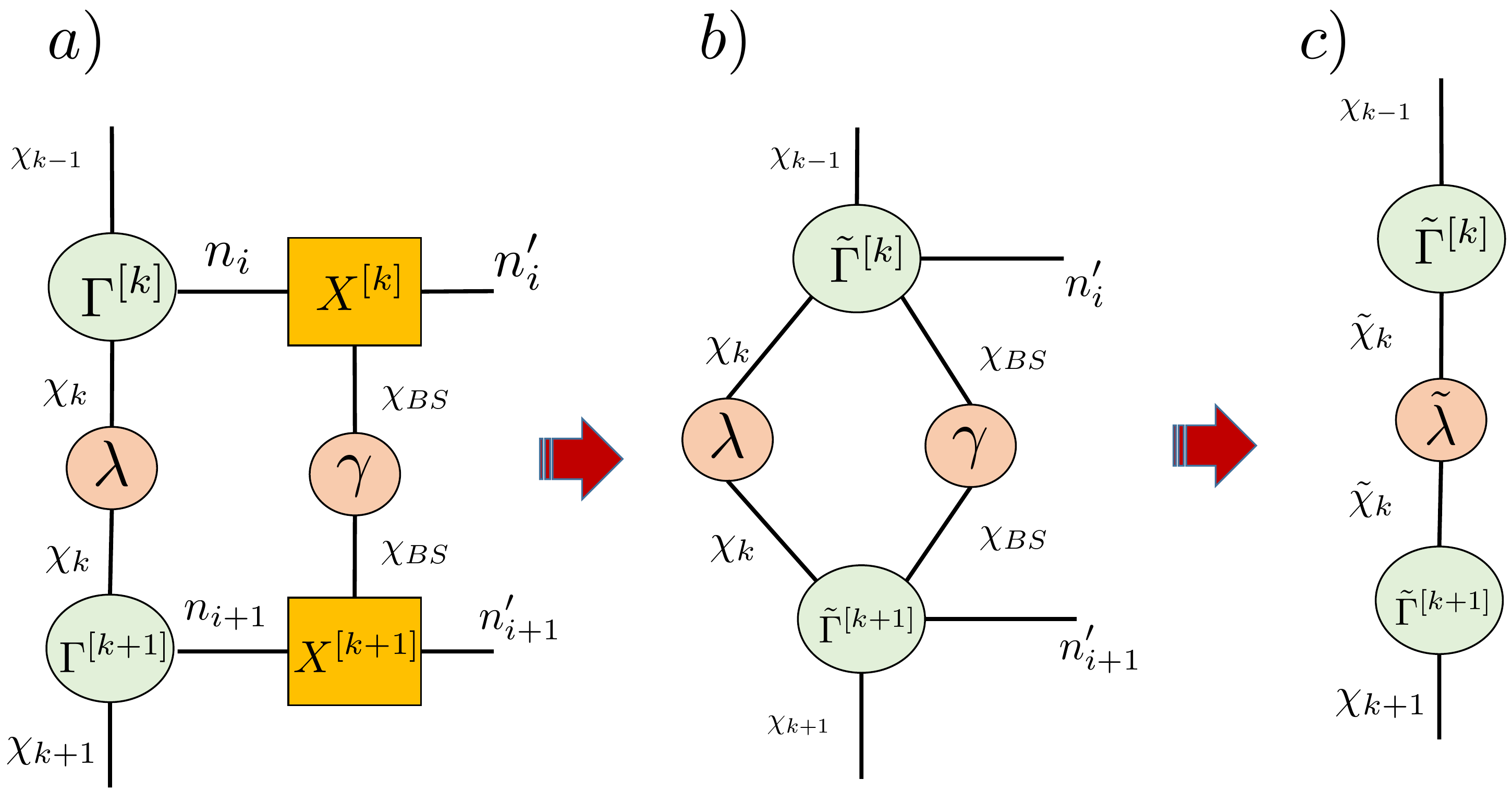}
\caption{a) The action of a two-mode coupler on modes $k$ and $k+1$ of a matrix product state written on its canonical form is first obtained by
doing a singular-value decomposition of the matrix product operator of the coupler;
b) Secondly, we contract the tensors $\Gamma^{[k]}$ and $X^{[k]}$ of mode $k$ and
$\Gamma^{[k+1]}$ and $X^{[k+1]}$ of mode $k+1$, giving $\tilde{\Gamma}^{[k]}$ and  $\tilde{\Gamma}^{[k+1]}$ respectively;
c) Finally, we relabel the two singular values $\lambda$ and $\gamma$ into a new label $\tilde{\lambda}$ of the resulting
matrix product state.}
\label{fig:figure5}
\end{figure*}
A two-mode coupler $B^{[k,k+1]}$ acting on modes $k$ and $k+1$ is modeled by a 4 legs tensor, i.e., a matrix product operator,
with physical  indexes $n_k$ and $n_{k+1}$ for the input and $n'_k$ and $n'_{k+1}$ for the output,
\begin{widetext}
\begin{equation}
B^{[k,k+1]}=\sum_{n_k,n_{k+1},n'_k,n'_{k+1}} C_{n_k,n_{k+1}}^{n'_k,n'_{k+1}}\ket{n'_k,n'_{k+1}}\bra{n_k,n_{k+1}},
\end{equation}
\end{widetext}
where the coefficients $C_{n_k,n_{k+1}}^{n'_k,n'_{k+1}}$ are the well-known input-output amplitudes of a beamsplitter
\cite{CaSaTe89,KiSoBuKn02} (see also equation~(3.9) in \cite{Aaronson2013}).

In \cite{Temme2012} an algorithm was constructed based on directly applying the unitary $B^{[k,k+1]}$ to the matrix product state
followed by a singular-value decomposition to rebuild the canonical form of the output state, reaching a 
computational cost scaling as $\text{O}(\chi^3d^3)$.
In what follows we present an alternative algorithm that provides a better scaling when the bond dimension $\chi$ is higher than the physical dimension $d$,
which is generally the case in most realistic simulations.

As shown in Figure \ref{fig:figure5} a),
in order to model the evolution of modes $k$ and $k+1$ under a two-mode coupler operation, we first implement a singular-value decomposition of the
matrix product operator $B^{[k,k+1]}$ with respect to the separation between $(n_k,n'_k)$ and $(n_{k+1},n'_{k+1})$ indexes, i.e.,
\begin{equation}
B^{[k,k+1]}_{(n_k,n'_k),(n_{k+1},n'_{k+1})}=\sum_{\gamma=1}^{\chi_{BS}}X^{[k]}_{n_k,n'_k,\gamma}\sigma^{[k]}_{\gamma}X^{[k+1]}_{n_{k+1},n'_{k+1},\gamma}
\end{equation}
where $\chi_{BS}$ is the Schmidt rank of the matrix product operator.
The Schmidt rank of a singular-value decomposition of a matrix being upper-bounded by the
largest of the two local dimensions provides the bound
\begin{equation}
 \chi_{BS}\leq (d+1)^2,
 \label{eq:upperBSchi}
\end{equation}
where the running time of the matrix product operator decomposition scales as $\text{O}\left((d+1)^6\right)$.

As shown in Figure \ref{fig:figure5} b), the next step is to contract the tensors $\Gamma^{[k]}$ and $X^{[k]}$ of mode $k$ and
$\Gamma^{[k+1]}$ and $X^{[k+1]}$ of mode $k+1$ in order to generate the tensors $\tilde{\Gamma}^{[k]}$ and  $\tilde{\Gamma}^{[k+1]}$ of the
state resulting after the beamsplitter transformation.

The running time of the contraction leading to the tensor $\tilde{\Gamma}^{[k]}$ scales as
$\chi_{k-1}\chi_{k}\chi_{BS}(d+1)^2$ where for
$\tilde{\Gamma}^{[k+1]}$ scales as $\chi_{k+1}\chi_{k}\chi_{BS}(d+1)^2$
which leads to a scaling of the contraction running time
\begin{equation}
T_C=\text{O}\left(\chi^2\chi_{BS}(d+1)^2\right)\leq \text{O}\left(\chi^2(d+1)^4\right).
\end{equation}

We remark that, as shown in Figure \ref{fig:figure5} c), the tensors $\tilde{\Gamma}^{[k]}$ and  $\tilde{\Gamma}^{[k+1]}$
are connected by two pairs of singular values,
$\chi_k$ from the initial state and $\chi_{BS}$ from the beamsplitter matrix product operator, which can be merged into
a single $\tilde{\chi}$ satisfying $\tilde{\chi}_k=\chi_k\chi_{BS}$, which
combined with equation (\ref{eq:upperBSchi}) provides the bound
\begin{equation}
\tilde{\chi}\leq\chi(d+1)^2,
\label{eq:scalingchi}
\end{equation}
which is the equivalent of Lemma 4 (i) in \cite{Jozsa06}.

\subsubsection{Circuit simulation}
The ideal boson sampling input state corresponds to a trivial matrix product state of bond dimension $\chi=1$,
composed of $N$ tensors $\Gamma^{[i]}_1=\delta_{n_i,1}\delta_{\alpha_{i-1},0}\delta_{\alpha_{i},0}$ encoding single-photon inputs
and $M-N$  tensors $\Gamma^{[i]}_0=\delta_{n_i,0}\delta_{\alpha_{i-1},0}\delta_{\alpha_{i},0}$ encoding vacuum inputs.
For every layer of couplers we apply in parallel the matrix product update detailed in subsections \ref{subsec:1mode} and \ref{subsec:2mode}. The bond dimension scales with the depth of the circuit $D$ as
$\text{O}\left((d+1)^{2D}\right)$, the storage space for the tensors as $\text{O}\left(M(d+1)^{4D+1}\right)$,
and the computational cost of the contraction of the matrix product state scales as $\text{O}\left(M(d+1)^{4(D+1)}\right)$,
where the leading order corresponds to the contractions of the last layer of gates.

\subsubsection{Sampling from a matrix product state}

Once the matrix product state resulting from $D$ layers of gates has been calculated, it is well known that one can  exploit  the chain rule
\begin{equation}
p(n_{1}, \ldots, n_{M})=p(n_{M} | n_{M-1}, \ldots, n_{1}) \ldots p(n_{1})
\end{equation}
to generates samples of $p(\bar{n})$. For completeness, we reproduce the explanation 
in \cite{Lubasch18} bellow. 

First calculate for each of the $d+1$ outcomes $n_1$ the probability 
${\rm Tr}\left[\proj{\psi}\proj{n_1}\otimes I_{2...M}\right]$,
where $\proj{n_1}$ is the projector onto the photon number state $n_1$ of mode 1 and $I_{2...M}$ is the identity operator on modes 2 to $M$. 
This is done by contraction of the matrix product state with itself, 
interleaved with a matrix product operator representing the measurement projector $\proj{n_1}$.
Then we randomly select one of the $d+1$ potential outcomes $n_1$ and
update our state by generating $|\psi_{\tilde{n}_{1}}\rangle := \langle \tilde{n}_{1}|\psi \rangle$, where the bra $\langle \tilde{n}_{1}|$
acts only on mode $1$.
The result of this contraction is a new, unnormalized matrix product state $|\psi_{\tilde{n}_{1}}\rangle$ of size $N-1$.
Note that this new matrix product state satisfies the condition $\langle \psi_{\tilde{n}_{1}} | \psi_{\tilde{n}_{1}}\rangle = p(\tilde{n}_{1})$.
The second step now uses the state $|\psi_{\tilde{n}_{1}}\rangle$.
Firstly, we calculate the $d+1$ outcome probabilities
$p(n_{2}, \tilde{n}_{1}) := \langle \psi_{\tilde{n}_{1}} | \left( |n_{2}\rangle \langle n_{2}| \otimes \mathcal{I}_{3, \ldots, N} \right) |\psi_{\tilde{n}_{1}} \rangle$
and randomly select a $\tilde{n}_{2}$ from the probability distribution $p(n_{2} | \tilde{n}_{1}) := p(n_{2}, \tilde{n}_{1}) / p(\tilde{n}_{1})$.
Secondly, we generate a new, unnormalized matrix product state $|\psi_{\tilde{n}_{1}, \tilde{n}_{2}}\rangle := \langle \tilde{n}_{2}|\psi_{\tilde{n}_{1}} \rangle$ of size $N-2$.
Continuing this procedure for the remaining $M-3$ output modes, we end up with one sample drawn according to the probability distribution $p(n_{1}, n_{2}, \ldots, n_{N})$.
The highest computational cost corresponds to the contraction leading to $p(n_{1}, n_{2}, \ldots, n_{N})$.
A trivial contraction algorithm provides a running time of $\text{O}\left(M\chi^4(d+1)^2\right)$, which for matrix product state resulting from $D$ layers of couplers becomes $\text{O}(M^2(d+1)^{8D+2})$.

\subsection{Simulating ideal logarithmic depth circuits}
It is important to notice that the bond dimension scales exponentially with the depth of the circuit.
If $d$ was a constant, such as in spin systems simulations, a shallow circuit satisfying a logarithmic depth constraint as in
eq.~(\ref{eq:circuitthreshold}) would lead to a polynomial time algorithm. In a tensor network simulation of quantum optics,
the potential bunching of photons, which can all potentially accumulate in a given mode, makes the simulation harder.
In order to obtain an exact simulation of the evolution and the sampling of $N$ input single photons over a circuit of depth $D$
we fixed the physical dimension over the whole evolution to $d=N$, the total number of photons.
Because $N$  scales with the number of modes $M$, see eq.~(\ref{eq:photonscaling}),
the computational cost of contraction, storage and sampling becomes quasipolynomial in the size of the system.

\section{Generalization to non-uniform losses}
\label{sec:nonuniformlosses}

In subsection \ref{subsec:modelloss} we have shown that a lossy linear optics interferometer is mathematically modeled by a complex linear transformation $A$ with singular value decomposition $A=V\hat{\mu} W$.
As shown in Figure \ref{fig:figure2}, this is equivalent to
a circuit composed by a lossless linear optics transformation $V$, followed by $M$ parallel set of different pure-loss channels of transmission $\mu_i$, and a final lossless linear optics transformation $W$.

\begin{figure}
\centering
  \includegraphics[width=.9\linewidth]{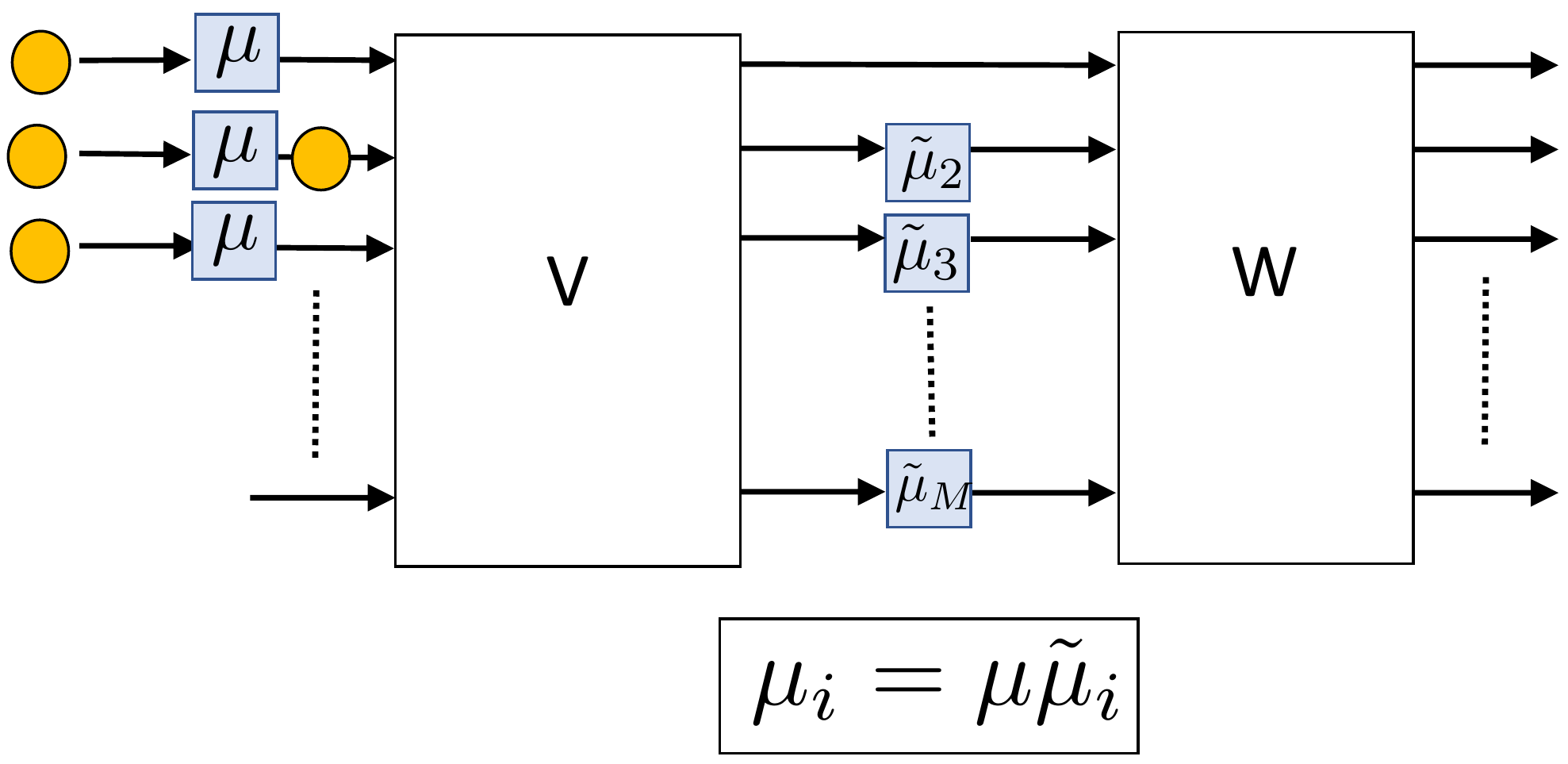}
\caption{A non-uniform lossy interferometer is equivalent to $M$ parallel set of pure-loss channels of transmission $\mu=\max{\mu_i}$
(where we have chosen $\mu=\mu_1$ without loss of generality) followed by the losseless circuit $V$, $M$ parallel set of pure-loss channels ($\tilde{\mu}_i$) and a final losseless circuit $W$.}
\label{fig:figure6}
\end{figure}

As shown in  Figure \ref{fig:figure6}, each of the $M$ parallel pure-loss channels of transmission $\mu_i$ can be decomposed into a concatenation a pure-loss channels of transmission $\mu=\max{\mu_i}$  followed by one of 
transmission $\tilde{\mu}_i=\mu_i/\mu$.
Because $M$ parallel pure-loss channels of transmission $\mu$ commute with the unitary $V$, we obtain a scheme
where $M$ parallel set of pure-loss channels of transmission $\mu$ are followed by the ideal circuit $V$, $M$ parallel different pure-loss channels ($\tilde{\mu}_i$) and a final ideal circuit $W$.
The state $\rho_{\text{in}}$ now results from applying $M$ parallel set of pure-loss channels of transmission $\mu=\max_i\mu_i$
and the thermal state $\rho_{\text{T}}$ is chosen such that $\lambda=\mu$, in the same way we did in subsection \ref{sec:mainresult}.
The only difference is that now the set of operations after the pure-loss channels of transmission $\mu$ is not longer a lossless interferometer
$U$ but a $M$-mode quantum channel $\mathcal{L}$ resulting from concatenating $V$, $M$ parallel pure-loss channels ($\tilde{\mu}_i$) and $W$.

\subsection{Generalizing Theorem \ref{th:losses}}
Because applying a quantum operation $\mathcal{L}$ can only make the trace-distance decrease, similarly as for a measurement,
it is trivial to see that one can generalize Theorem \ref{th:losses} by replacing the uniform losses $\mu$ by $\mu=\max{\mu_i}$.

\subsection{Generalizing lemma \ref{lem:classicalalgo}}
It is a well-known fact in quantum optics
that the action of a pure loss channel of transmission $\mu_i$ on a coherent state $\ket{\alpha}$ outputs a weaker coherent state
$\ket{\mu_i\alpha}$. Therefore, the evolution of an input multimode coherent state $\ket{\bm{\alpha}}$ can be easily computed,
by implementing the matrix multiplication $\bm{\beta}=A\bm{\alpha}$.
Once the output coherent states have been determined, the sampling from their respective Poisson distributions proceeds as before.

\subsection{Generalizing lemma \ref{lem:TN}}
The adaptation of the tensor network simulation is slightly more involved. Let's use the notation $A_{i,j}$
for the coupler acting on modes $i$ and $i+1$ at the layer of couplers $j$. Every $A_{i,j}$ has a decomposition into
a unitary $V_{i,j}$ followed by two independent pure-loss channels and a final unitary $W_{i,j}$. 
Because every pure-loss channel can be seen
as a lossless coupling interaction with an environmental mode, it is easy to see that a circuit with losses can be transformed into
an ideal lossless circuit by doubling the number of couplers and adding two ancillary modes per coupler with losses.
We can then place all the ancillary modes interacting with mode $i$ between input modes $i$ and $i+1$, i.e.,
$D$ of them bellow each input mode for a circuit of depth $D$.
For a lossy circuit of depth $D$ there is at most $3D$ 
lossless gates acting on each mode with a range of at most $D$.
As detailed in \cite{Jozsa06}, one can transform a $D$ range gate into $2D$ nearest-neighbor gates.
Therefore, our initial circuit with losses of depth $D=M$ becomes a 
lossless circuit with $M^2$ modes and $6M^2$
nearest-neighbor gates. This leads to a less favorable scaling of
the computational cost of contraction, storage and sampling, but which remains quasipolynomial in $M$.
This last algorithm is certainly not optimal and we are convinced that more elaborate choices can certainly improve the simulation of multi-photon interference with non-uniform losses.

\section{Conclusion}
\label{sec:conclusion}

The vast majority of currently proposed boson sampling architectures
suffer from exponential decay of the transmission with the length of the circuit.
We have shown that multi-photon interference over an $M$ modes interferometer of depth $D$
can be efficiently simulated classically.
More precisely, we have show that either the depth of the circuit is large enough 
($D\geq \text{O}(\log M)$) that it can be simulated by thermal noise with an algorithm running in polynomial time, or the depth of the circuit is shallow enough ($D\leq \text{O}(\log M)$) that a tensor network simulation runs in quasi-polynomial time.

We also showed that for even very optimistic experimental parameters,
a quantum supremacy experiment using the current boson sampling hardness proof is out of reach.
We believe that our result suggests that in order to implement a quantum advantage experiment with single-photons and linear optics we need novel theoretical ideas or radically new technological developments. One possibility would be to shift to platforms with very low algebraic 
transmission where our result would not be applicable. 
Another option would be prove the hardness of novel boson sampling architectures beyond the planar circuit architecture. A promising route would be to reduce the depth of the circuit to the shallow regime while maintaining the complexity by moving to a lattice structure.
A potential candidate would be an adaptation to quantum optics of the recent proposal of  quenching a spin lattice \cite{Bermejo2017}.

We discussed that the potential bunching of photons makes the tensor network simulation of quantum optical systems more involved than its finite spin counterparts.
One could potentially restore the polynomial scaling of logarithmic depth circuits observed for finite systems by designing an $\epsilon$-approximate algorithm that truncates every mode to a finite size. To our knowledge, this is a non-trivial result that is certainly worth pursuing in future research.

One of the motivations of our work was to show that simulating boson sampling with imperfections is indeed easier than ideal boson sampling. We achieved this goal for a restricted regime of losses where the system becomes classically simulatable. We conjecture that there should exist a family of algorithms that optimally interpolate between the Clifford and Clifford algorithm for ideal devices and fully efficient algorithms for noisy devices. 
This work is only a first step in this direction and we believe that further results will
improve even further the classical simulation of imperfect boson sampling devices. 

An interesting open question is whether our proof can be adapted to other technological platforms candidates to a quantum advantage test.

It has been recently proposed to use post-selection to improve boson sampling rates while maintaining its hardness \cite{Wang18}. The proposal, inspired by the fact that the hardness of boson sampling is preserved when only a constant number of photons is lost \cite{Aaronson2016}, proposes to post-select outcome events with a constant number of lost photons. If the probability of a successful
post-selection decreases exponentially with the increasing size of the system, 
it was argued that it could be a useful tool for experiment of the limited size needed for a practical quantum advantage demonstrations. If the proof presented in this manuscript can not be used to discard this novel approach to boson sampling, a recent work by the same authors \cite{Renema18} imposes strong constrains on its viability.

After the completion of this article we learned about Ref. \cite{Oszmaniec18}
that obtains a similar result to our Theorem \ref{th:losses} and its generalization to non-uniform losses using very different techniques.

\section{Acknowledgements}
We would like to thank Anthony Leverrier for fruitful discussions on the constellations of coherent states.
R.G.-P. is Research Associate of the F.R.S.-FNRS. J.J.R. acknowledges NWO Rubicon.
V.S.  was supported by the National Council for Scientific and Technological Development (CNPq) of Brazil,
grant  304129/2015-1, and by  the S{\~a}o Paulo Research Foundation   (FAPESP), grant 2015/23296-8.
R.G.-P. and J.J.R. acknowledge funding from the Fondation Wiener-Anspach.
This project has received funding from the European Union’s Horizon 2020 research and innovation program under grant agreement No 641039 (QUCHIP).

\end{document}